\documentclass{amsart}
\usepackage{amsmath, amsthm, amssymb,fancyhdr}
\usepackage{amsmath,amssymb,amsbsy, amsthm, amsfonts}
\usepackage{url}
\usepackage{pgf}
\usepackage{xxcolor}
\usepackage{xcolor}

\usepackage[pdftex, colorlinks=true, linkcolor=blue, urlcolor=blue, citecolor   = red]{hyperref}

\usepackage{ amsmath,amssymb,amsbsy,amsfonts, latexsym,amsopn,amstext, amsxtra,euscript,amscd, amsthm, amsfonts}

\usepackage{graphics,wrapfig,times}

\usepackage[author-year, msc-links]{amsrefs}

\issueinfo{10}{1}{number 1}{2016}
\copyrightinfo{2016}{\href{http://albanian-j-math.com}{Albanian Journal of Mathematics}}
\def\issn{{\sc \textbf{ISSN: }} 1930-1235; }
\def\issueyear{\textbf{2016}}
\PII{\issn  (\issueyear)}

\pagespan{37}{45}

\title{Graph based linear error correcting codes}


\usepackage{amsrefs}
\usepackage{tikz}

\newtheorem{theorem}{Theorem}[section]

\newtheorem{proposition}[theorem]{Proposition}

\begin{document}

\maketitle

\hrule

\vspace{.5cm}

\begin{center}

\sc{Monika Polak} 

\smallskip

\textit{University of Maria Curie--Sklodowska\\
Lublin, Poland \\
Email: monika.katarzyna.polak@gmail.com 
}

\bigskip

\sc{Eustrat Zhupa}

\smallskip

\textit{University of Maria Curie--Sklodowska\\
Lublin, Poland \\
Email: e.zhupa@gmail.com
}

\end{center}

\vspace{.3cm}

\hrule

\vspace{.2cm}

\begin{abstract}
In this article we present a construction of error correcting codes, that have representation as very sparse matrices and belong to the class of Low Density Parity Check Codes. LDPC codes are in the classical Hamming metric. They are very 
close
to well known Shannon bound. The ability to use graphs for code construction was first discussed by 
Tanner in 1981 and has been used in a number of very effective implementations.  We describe how to  construct such codes by using special a family of graphs introduced by Ustimenko and Woldar. Graphs that we used
are bipartite, bi-regular, very sparse and do not have short cycles $C_4$. Due to the very low density of such 
graphs, the obtained codes are fast decodable. We describe how to choose parameters to obtain a desired code rate. We also show results of computer simulations of BER (bit error rate) of the obtained codes in order to compare them with other known LDPC codes.
\end{abstract}

\section{Introduction}

All information in a computer is represented as a sequence of binary digits.
Data are stored and shared with others. Both in the case of data storage and during data transfer, we need protection against transmission errors or data loss. 
In the first scenario, unreliable or faulty hardware (computer memories, compact discs, QR Code) can seriously corrupt data. Coding techniques are one of the important measures for improving reliability.

As to the second scenario, very often digital data are sent over unreliable communication channels (air, a telephone line, a beam of light or a cable). Because of the channel, noise errors may be introduced during transmission from the source to the receiver. It is very important for the recipient to receive the correct message as intended. 
In order to minimize the number of errors during transmission, we can use error correcting codes.

 Coding of information with the use of error correcting codes
consists of adding to sequences of $K$ elements some extra bits in a
certain way. Such additional bits don't carry any information and they only have
check purposes. An error correcting code is $A  \subset \mathbb F^N_2$ , where
$\mathbb F_2 = \{0, 1\}$ and codewords follow the classical Hamming metric:
\[d(x, y) = |\{i : x_i \neq y_i \}|.\]
We denote with $[N, K]$ the code with code words length $N$ and
$K$ information bits. In such a code there are $R = N - K$ parity
check equations. The ratio $K/N$ is called \emph{code rate} and is denoted by $R_C$. It is
interesting to look at codes with the best correction properties and
the biggest code rate for cost reasons.

Linear error correcting code can be represented in three ways: by the
generator matrix $G$, by parity check matrix $H$ or by Tanner graph $\Gamma(V,E)$. Parity check matrix $H$ for $[N, K]$ code is $R \times N$ matrix whose words
are zeros or ones. Rows of such matrix correspond to the parity
checks and the columns correspond to codeword bits. If bit number $j$ in the
codeword is checked by parity check number $i$ then in position
$(i , j)$ in matrix $H$ there is a $1$, if not there is a $0$.
Switching column does not change code properties and provides an
equivalent code. A simple example of linear error correcting code is Hamming code $[7,4]$ with matrix $H$ of the form:\\

\[
H=
 \begin{pmatrix}
1 & 1 & 1 & 0 & 1 & 0 & 0 \\
1 & 0 & 1 & 1 & 0 & 1 & 0 \\
1 & 1 & 0 & 1 & 0 & 0 & 1
\end{pmatrix}
\]
To encode vector of 4 bits by using $[7,4]$ Hamming code we add 3 extra bits.
Each bit is checked by a unique set of control equations. We assume that every codeword is from the set
\[C = \{y \in \mathbb F_2^N : Hy^T = 0\}.\]

Let's recall few simple facts from Graph Theory that can be found for example in \cite{biggs}. 
The \emph{distance} between vertices $v_1$ and $v_2$ of the graph is the length
of minimal path from $v_1$ and $v_2$. The graph is connected if for
arbitrary pair of vertices $v_1$, $v_2$ there is a path from $v_1$ to $v_2$. The
\emph{girth} of simple graph is the length of the shortest cycle in graph. A \emph{bipartite} graph is a graph
$\Gamma(V, E)$, in which a set of nodes $V$ can be divided into two subsets
$V = V_1 \cup V_2$ ($V_1 \cap V_2=\emptyset$) in such a way that no two vertices from each set $V_i$ ,
$i = 1,$ 2 are connected with an edge. We refer to bipartite graph $\Gamma(V, E)$
with partition sets $V_i$, $i = 1$, 2, $V = V_1 \cup V_2$ as bi-regular one if the
number of neighbours for representatives of each partition set are
constants $s$ and $r$ (bi-degrees); we say that we have bi-regularity $(s,r)$. We call the graph \emph{regular} in the
case $s = r$. 

 Tanner in 1981 introduced an effective graphical representation for LDPC Tanner codes, i.e. Tanner graph.
\textsl{Tanner graph} is the bipartite graph in which one subset $V_1$ corresponds
to the codeword bits and second $V_2$ to the parity checks. Vertex from the set
$V_1$ is connected to a vertex from the set $V_2$ if and only if a bit corresponding to
the vertex from $V_1$ is controlled by the parity check corresponding to the vertex
from $V_2$.

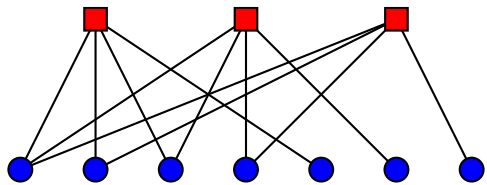
\begin{figure}[h!]
\begin{center}
\begin{tikzpicture}[style=thick,rotate=90,scale=0.50]
\foreach \x in {12}
\foreach \y in {2, 6, 10} {
\draw (4, \y) -- (0, \x);}
\draw (4, 10) -- (0, 10);
\draw (4, 2) -- (0, 10);
\draw (4, 10) -- (0, 8);
\draw (4, 6) -- (0, 8);
\draw (4, 6) -- (0, 6);
\draw (4, 2) -- (0, 6);
\draw (4, 10) -- (0, 4);
\draw (4, 6) -- (0, 2);
\draw (4, 2) -- (0, 0);
\foreach \x in {0, 2, 4, 6, 8, 10, 12} {
  \draw[fill=blue] (0, \x) circle (9pt);
}
\foreach \y in {2, 6, 10} {
  \draw[fill=red] (4 - 0.3, \y - 0.3) rectangle (4 + 0.3, \y + 0.3);
}
\end{tikzpicture}
\caption{Tanner graph for $[7,4]$ Hamming code}
\end{center}
\end{figure}

Tanner graph represents parity check equations.
There is a standard way to create error correcting code codes depending on adjacency matrix of bipartite, bi-regular graph. 
Parity check matrix $H$ is a part of the adjacency matrix $A$ of the graph with desired properties
used to create the code:
\[
A=
 \begin{pmatrix}
0 & H \\
H^T & 0 \\
\end{pmatrix}
\]
Determination of the matrix $H$ is equivalent to code designation. However parity check matrix is not unique.
Switching
columns doesn't change code properties and gives us an equivalent code.

One famous class of error correcting codes is Low Density Parity Check Codes introduced in 1962 by
Robert G. Gallager \cite{gallangerr:62}. Such codes have a wide range for
selection of parameters, enabling creation of codes with a
large block size and excellent correction properties. In the
present work we show new results on applications of Computer
Algebra and Theory of Algebraic Graphs in constructions of new
LDPC error correcting codes. Only specific graphs are suitable for
such purpose. Usually, simple graphs are used. The graph should be
bipartite, sparse, without small cycles and bi-regular or regular with
the possibility to obtain bi-regularity. Codes related to graph
defined below have all such important properties.

A code with a representation as a sparse matrix or a sparse Tanner
graph is a Low-Density Parity-Check Code, \cite{gallangerr:62}. A matrix is sparse if number of ones in it is small compared to number of zeros.
Low Density Parity Check Codes have a very sparse parity check
matrix. A \emph{sparse} graph has a small number of edges in relation to the
number of vertices. A simple relationship describing the density of
the graph $\Gamma(V, E)$ is
\begin{equation}\label{density}
D=\frac{\Large2|E|}{\Large|V|(|V|-1)},
\end{equation}
where $|E|$ is the number of edges of graph $\Gamma$ and $|V|$ the number
of vertices.

\section{Description of the graph}

Let $q$ be a prime power and let $\mathbb F_{q^2}$ be the quadratic extension of $\mathbb F_q$. The family of graphs $F=F(\mathbb F_q,\mathbb F_{q^2})$ was introduced in \cite{woldar:03}.  In fact, the described graphs are affine part 
of generalized quadrangles.
The following representation can be found in \cite{skrytustimenko}, where they are denoted as $I4_q$.
Those graphs are bipartite with set of vertices $V=V_1\cup V_2$, where $V_1\cap V_2=\emptyset$.
They have girth at least 8  
and very different bi-regularity $(q,q^2)$.
Traditionally because of geometric construction, one partition set
$V_1 = P$ is called set of points and other $V_2 = L$ is called the set of
lines:
\[P=\{(a,b,c):a\in \mathbb F_{q}, b\in \mathbb F_{q^2}, c\in \mathbb F_{q}\},\]
\[L=\{(x,y,z):x\in \mathbb F_{q^2}, y\in \mathbb F_{q^2}, z\in \mathbb F_{q}\}.\]
Two types of brackets are used to distinguish points and lines.
Let $x \longrightarrow x^q$ be the
Frobenius automorphism of $\mathbb F_{q^2}$. 
We say point $(p)$ is \emph{incident} to line $[l]$ in graph $F(\mathbb F_q,\mathbb F_{q^2})$, and we
write $(p)I [l]$, if the following relations on their coordinates hold:
\begin{equation}\label{1}
\left\{
  \begin{array}{ll}
    y-b=ax\\
    z-c=ay+ay^q\\
  \end{array}
\right.
\end{equation}
The set of vertices is $V(F)=P\cup L$ and the set of edges consists of all pairs $((p),[l])$, for which $(p)I[l]$.
Because $a\in \mathbb F_{q}, b\in \mathbb F_{q^2}, c\in \mathbb F_{q}, x\in \mathbb F_{q^2}, y\in \mathbb F_{q^2}, z\in \mathbb F_{q}$ we have $|P|=q^4$ and $|L|=q^5$ ($|V(F)|=q^5+q^4=q^4(q+1)$). This is a family of simple graphs.
\begin{proposition}\label{prop1}
$F(\mathbb F_q,\mathbb F_{q^2})$ is a family of sparse graphs.
\end{proposition}
\begin{proof}
If we set point $(p)=(a,b,c)$ and $x$ coordinate of line $[l]$ ($x\in \mathbb F_{q^2}$) then we can calculate $y$ and $z$ from \ref{1}. There is just one solution. So it's easy to see that each point has exactly $q^2$ neighbours. We have $|P|=q^4$ so set of edges has $|E|=q^4\cdot q^2=q^6$ elements.
Using formula \ref{density} the density of described graphs is 
\[D=\dfrac{2q^6}{q^4(q+1)(q^4(q+1)-1)}=\dfrac{2q^2}{2q^6+2q^5+q^4-1}\approx\dfrac{2}{2q^4+2q^3+q^2}.\]
\end{proof}

Instead of using elements of fields $\mathbb F_{q^2}$ and $\mathbb F_q$ as coordinates, we propose to use two rings $\mathbb Z_{n^2}$, $\mathbb Z_n$ and modulo operations. In such case sets $P$ and $L$
for the graph $F(\mathbb Z_n,\mathbb Z_{n^2})$ are the following :
\[P=\{(a,b,c):a\in \mathbb Z_{n}, b\in \mathbb Z_{n^2}, c\in \mathbb Z_{n}\},\]
\[L=\{(x,y,z):x\in \mathbb Z_{n^2}, y\in \mathbb Z_{n^2}, z\in \mathbb Z_{n}\}.\]
We say point $(p)$ is \emph{incident} to line $[l]$, and we write $(p)I [l]$, if the
following relations on their coordinates hold:
\begin{equation}\label{2}
\left\{
  \begin{array}{ll}
    (y-b)\mod n^2=(ax)\mod n^2\\
    (z-c)\mod n=(ay+ay^n)\mod n\\
  \end{array}
\right.
\end{equation}
Graphs defined in terms of finite rings as coordinates are bipartite,  
have girth at least 6 (probably 8, but not tested) and bi-regularity $(n,n^2)$. In this case they are not affine part of generalized quadrangles. Set of lines $L$ has $n^5$ elements and $|P|=n^4$. Set of vertices has $V=n^4(n+1)$ elements and set of edges consists of $n^6$ elements. By analogy with \ref{prop1}, the density of graphs $F(\mathbb Z_n,\mathbb Z_{n^2})$ is \[\dfrac{2}{nq^4+nq^3+n^2}.\]

\begin{figure}[h!]
    \begin{center}
    \includegraphics[scale =.3]{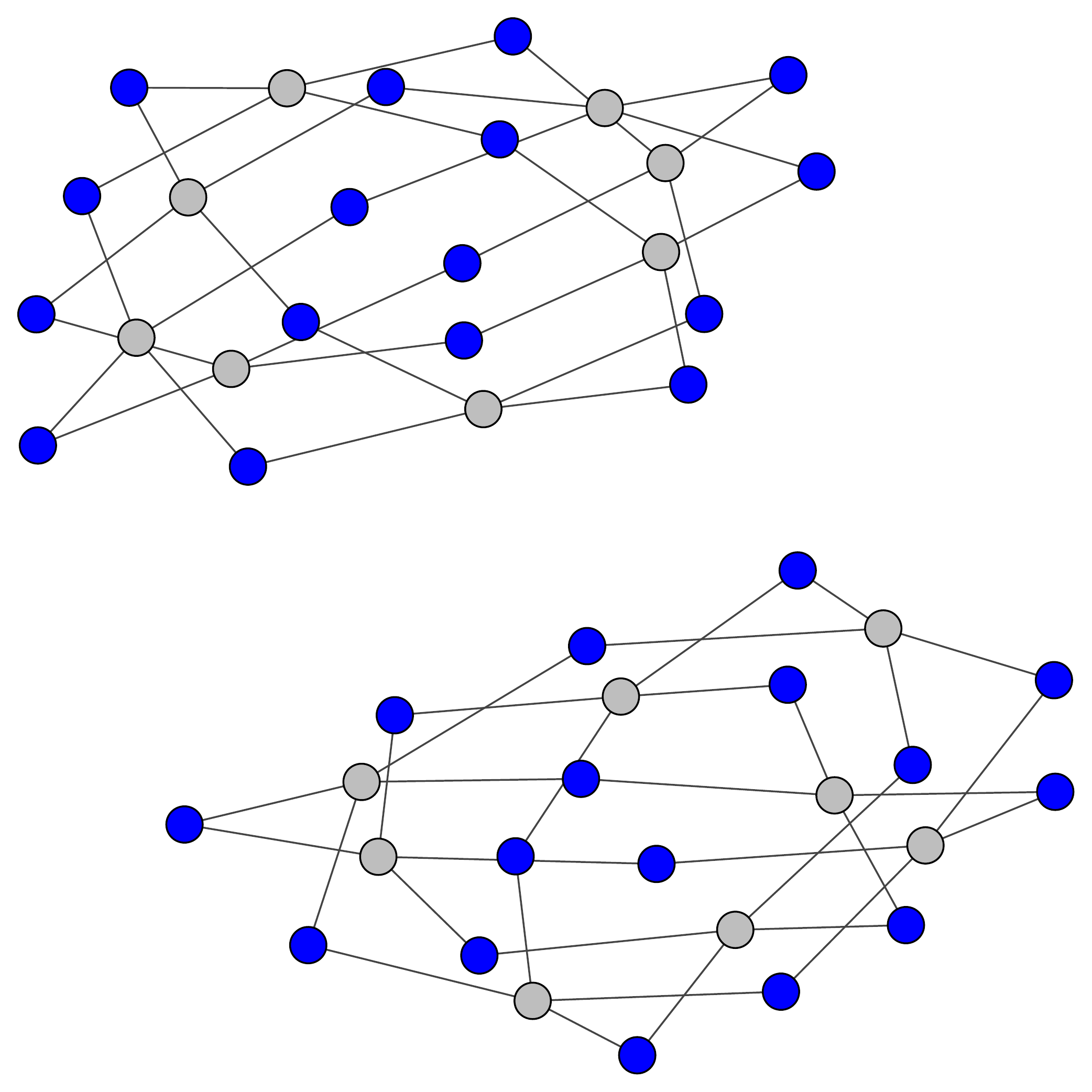}
    \caption{Graph $F(\mathbb Z_2,\mathbb Z_{4})$}
        \label{fig:0}
    \end{center}
\end{figure}

\begin{table}[h]
    \centering
    \caption{Properties of described graphs, for example number fields and finite rings}
    \vspace{0.1cm}
    \begin{tabular}{|c|cccc|}
    \hline
    Graph & Graph density & biregularity  & $|V|$ & $|E|$\\ \hline
    $F(\mathbb F_3,\mathbb F_9)$, $F(\mathbb Z_3,\mathbb Z_9)$    & $\approx 0.014$        & (3,9)  &324& 729   \\ 
    $F(\mathbb F_4,\mathbb F_{16})$, $F(\mathbb Z_4,\mathbb Z_{16})$    & $\approx 0.005$       & (4,16)  &1280& 4096     \\ 
    $F(\mathbb F_5,\mathbb F_{25})$, $F(\mathbb Z_5,\mathbb Z_{25})$    & $\approx 0.002$       & (5,25)  &3750&    15625   \\
	 $F(\mathbb Z_6,\mathbb Z_{36})$    & $\approx 0.001$       & (6,36)  &9072&    46656\\ 
    $F(\mathbb F_7,\mathbb F_{49})$, $F(\mathbb Z_7,\mathbb Z_{49})$    & $\approx 0.0006$   & (7,49)   &19208&  117649   \\ \hline
    \end{tabular}
    \label{tab:1}
\end{table}

Table \ref{tab:1} contains some properties of small representatives of described families. Figure \ref{fig:0} shows the smallest example: $F(\mathbb Z_2,\mathbb Z_{4})$ with $|P|=16$ (grey colour) and set of lines with $|L|=32$ elements. Each point has exactly 4 neighbours and each line has two neighbours. The advantages of using finite rings (and modulo operations) instead of prime powers is that graphs  $F(\mathbb Z_n,\mathbb Z_{n^2})$ are defined for all $n\geq2$.

\subsection{Code construction}


Our graphs $F$ are already biregular
but they also have a structure that enables us to remove points and lines
in such a way as to obtain a subgraph with different
bidegree. This operation yields a code with
better error correcting properties, but it reduces code rate and makes the code less convenient.
We can construct LDPC codes in two ways:
\begin{enumerate}
	\item In graphs $F(\mathbb F_q,\mathbb F_{q^2})$ or $F(\mathbb Z_n,\mathbb Z_{n^2})$z\ set of lines is bigger than set of points: $|L| > |P|$, so lines correspond to code words bits
and points correspond to parity checks. We decide to put one or zero in parity
check matrix by checking if relations ~\ref{1} between corresponding points and lines hold. 
Every bit of the codeword is checked by $q$ or $n$ parity checks. In this case parity check matrix $H$ is simply a part of adjacency matrix of used graph. $|L|=q^5$  and $|P|=q^4$ (or $|L|=n^5$ and $|P|=n^4$) so the size of $H$ is $q^4\times q^5$  ($n^4\times n^5$)  and \[R_C=\dfrac{q^5-q^4}{q^5}=\dfrac{q^5(q-1)}{q^5}=\dfrac{(q-1)}{q}\]
($R_C=\frac{n-1}{n}$).

\item Recalling that $L$ is the set of all lines and $P$ is the set of all points in
graph $F$, in order to obtain a bi-degree $(q, r )$ different from $(q, q^2)$, where $q<r<q^2$ we propose to put restrictions on the coordinates ($(n, r )$ different from $(n, n^2)$, where $n<r<n^2$) in the following way. Let $R \subset \mathbb F_{q^2}$ ($R \subset \mathbb Z_{n^2}$) be an $r$-element
subset and let $V_L$ be the set of lines in a new bipartite graph. This is the
following set:
\[V_L=\{[l]\in L| x\in R\},\]
where lines  $[l ]$  are represented by vectors $[x, y, z]$. Bi-degree reduction can only increase the
girth. After reduction, the bi-degree graph can be disconnected and
we use only one component to create a parity check matrix $H$.
In this case parity check matrix $H$ is a part of adjacency matrix of a subgraph of the used graph. $|V_L|=r\cdot q^3$  and $|P|=q^4$ (or $|V_L|=r\cdot n^3$ and $|P|=n^4$) so the size of $H$ is $q^4\times r\cdot n^3$  ($n^4\times r\cdot n^3$)  and 
\[R_C=\dfrac{r\cdot q^3-q^4}{r\cdot q^3}=1-\dfrac{q}{r}=\dfrac{(r-q)}{q}\]
($R_C=\frac{r-n}{n}$).

\end{enumerate}

\begin{table}[h]
    \centering
    \caption{Properties of described $[N,K]$ codes for example number fields and finite rings}
    \vspace{0.1cm}
    \begin{tabular}{|c|ccccc|}
    \hline
      Graph           & $N=|L|$ & $R=|P|$ & $K=N-R$ &$R_C$ & Girth  \\ \hline
		
 $F(\mathbb F_2,\mathbb F_4)$  & $2^5$ & $2^4$    &    16   & $0.5$ & $\geq8$  \\ 
	$F(\mathbb Z_2,\mathbb Z_4)$  & $2^5$ & $2^4$    &    16   & $0.5$ & $\geq6$  \\ 
 $F(\mathbb F_3,\mathbb F_9)$ & $3^5$ & $3^4$    &    162   & $\approx0.67$ & $\geq8$  \\ 
 $F(\mathbb Z_3,\mathbb Z_9)$ & $3^5$ & $3^4$    &    162   & $\approx0.67$ & $\geq6$  \\
 $F(\mathbb F_4,\mathbb F_{16})$  & $4^5$ & $4^4$    &    768   & $0.75$ & $\geq8$  \\
$F(\mathbb Z_4,\mathbb Z_{16})$ & $4^5$ & $4^4$    &    768   & $0.75$ & $\geq6$  \\
 $F(\mathbb F_5,\mathbb F_{25})$  & $5^5$ & $5^4$    &    2500   & $0.8$ & $\geq8$  \\ 
$F(\mathbb Z_5,\mathbb Z_{25})$ & $5^5$ & $5^4$    &    2500   & $0.8$ & $\geq6$  \\ 
$F(\mathbb Z_6,\mathbb Z_{36})$        & $6^5$ & $6^4$    &    6480   & $\approx0.83$ & $\geq6$  \\ 
	subgraph of $F(\mathbb Z_5,\mathbb Z_{25})$       &2000&625&1375&$\approx0.69$& $\geq6$ \\
  subgraph of $D(6,7)$                              & 2401&686&1715 & $\approx0.71$ & $\geq10$  \\ 
	subgraph of $D(8,5)$                              & 625 &250&375 & $0.6$ & $\geq12$  \\ 
	subgraph of $D(10,3)$                             & 243 &162&81 & $\approx0.33$ & $\geq14$  \\ \hline
		\end{tabular}
    \label{tab:2}
\end{table}

In first case $r=q^2$ and $s=q$ ($r=n^2$ and $s=n$). In the second  case  we choose  $r$ and $s=q$ ($s=n$).
In regular LDPC code every row has the same constant weight $r$
and every column has the same constant weight $s$. 
Despite the fact that graphs $F$ have very different bidegree, they give us regular LDPC codes. 
Every column of parity check matrix $H$ has the same weight.
\subsection{Corresponding LDPC codes}

Transmission quality depends mainly on code, on decoding algorithm
and level of noise in the communication channel. Code error
correcting properties are often tested by determining the
relationship between noise level and bit error rate. \emph{Bit error rate}
(BER) is the ratio of number of error bits to the total number of
transferred bits. The generated codes are based on the described graphs and  we performed computer simulation of a noisy channel. We encoded vectors, added White Gaussian Noise and decoded. The lack of short cycles guaranty convergence of the decoding algorithm. Presented LDPC codes work well with existing decoding algorithms, so there is no need do implement dedicated decoding technique.
Table \ref{tab:2} contains properties of described $[N,K]$ codes and other previously known codes in order to compare thein properties.
Fig. \ref{fig:1} and Fig. \ref{fig:2} present bit error rate for LDPC codes corresponding to described graphs. 
The code rates $R_C$ of such codes are bigger than $0.5$, so they are very convenient and have a 
small number of redundant bits. Tab. \ref{tab:2} contains properties of codes, for which results of the 
simulation are presented in this article, and parameters for a code based on graphs $D(n,q)$, $A(n,q)$ with similar code rate. 

Codes corresponding 
to graphs from special family $D(n,q)$  were presented in \cite{guinand2:97}. The presented codes have as good error correcting properties 
as codes constructed by Guinand and Lodge in \cite{guinand1:97}, whose construction is based on 
infinite family of graphs $D(n,q)$ constructed by Lazebnik and Ustimenko \cite{lazebnik:93}. 
In the case $q=p$ is a prime number the MATLAB code to generate LDPC codes is available in \cite{myPage}.

\begin{figure}[h]
    \begin{center}
    \includegraphics[scale =0.4,angle=-90]{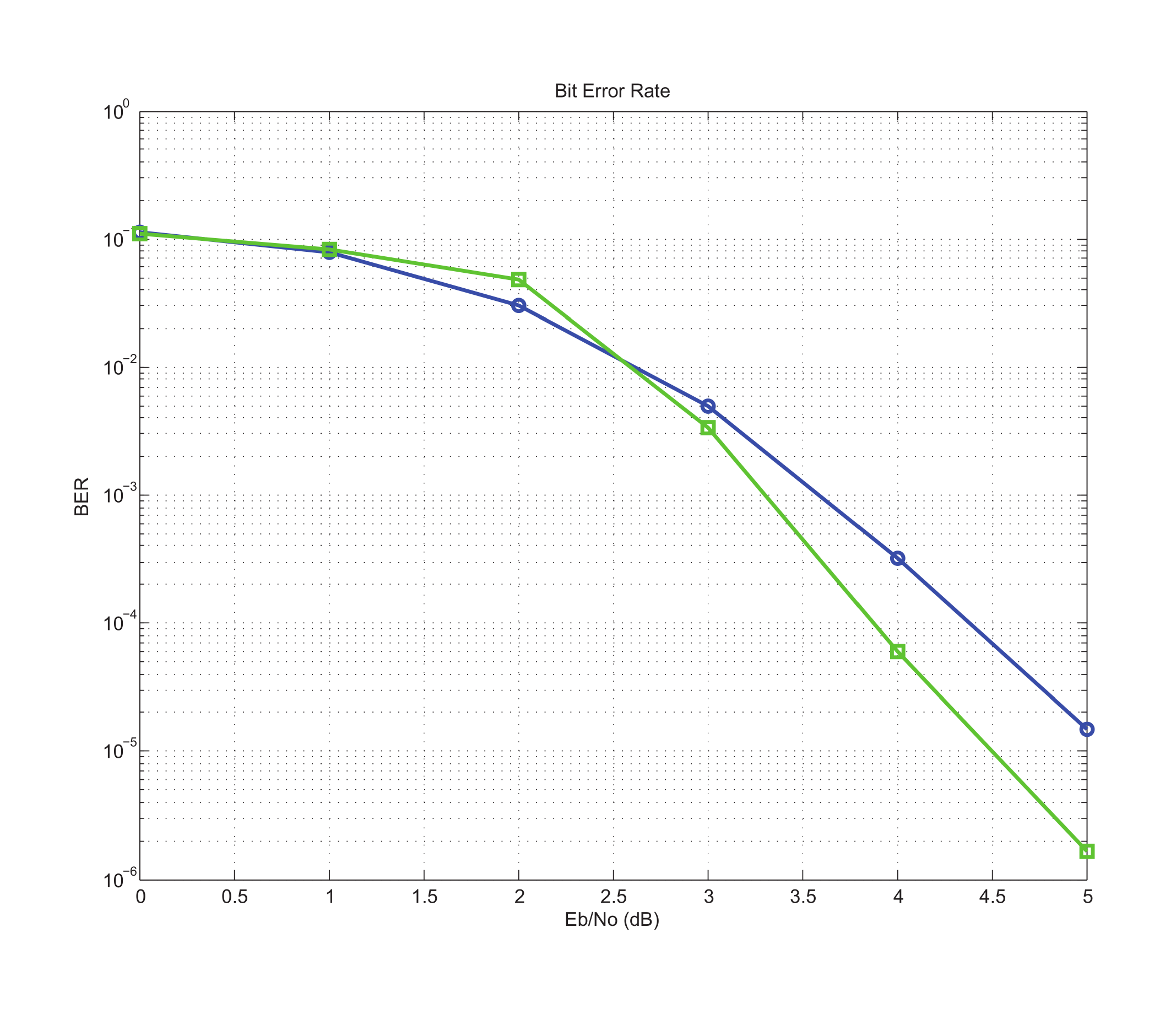}
    \caption{Bit error rate for $[243,162]$ code (blue) corresponding to graph $F(\mathbb F_3,\mathbb F_9)$ and $[1024,768]$ code (green) corresponding to graph $F(\mathbb F_4,\mathbb F_{16})$}
        \label{fig:1}
    \end{center}
\end{figure}

\begin{figure}[hb]
    \begin{center}
    \includegraphics[scale = .47,angle=-90]{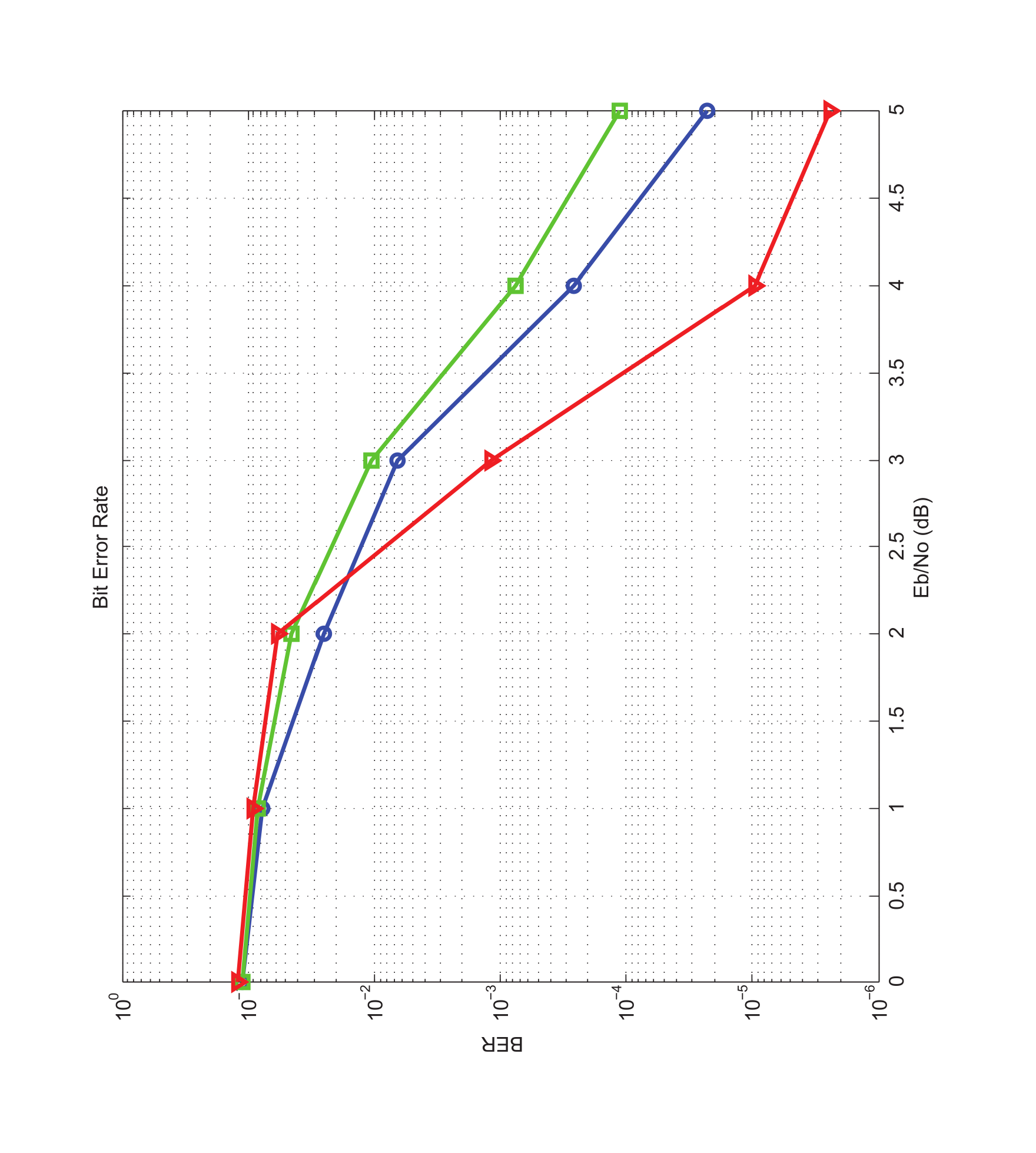}
    \caption{Bit error rate for $[243,162]$ code (blue) corresponding to graph $F(\mathbb Z_3,\mathbb Z_9)$, $[1024,768]$ code (green) corresponding to graph $F(\mathbb Z_4,\mathbb Z_{16})$ and $[2000,1375]$ code (red) corresponding to subgraph of $F(\mathbb Z_5,\mathbb Z_{25})$}
        \label{fig:2}
    \end{center}
\end{figure}

\clearpage

\begin{figure}[h!]
    \begin{center}
    \includegraphics[scale = .4]{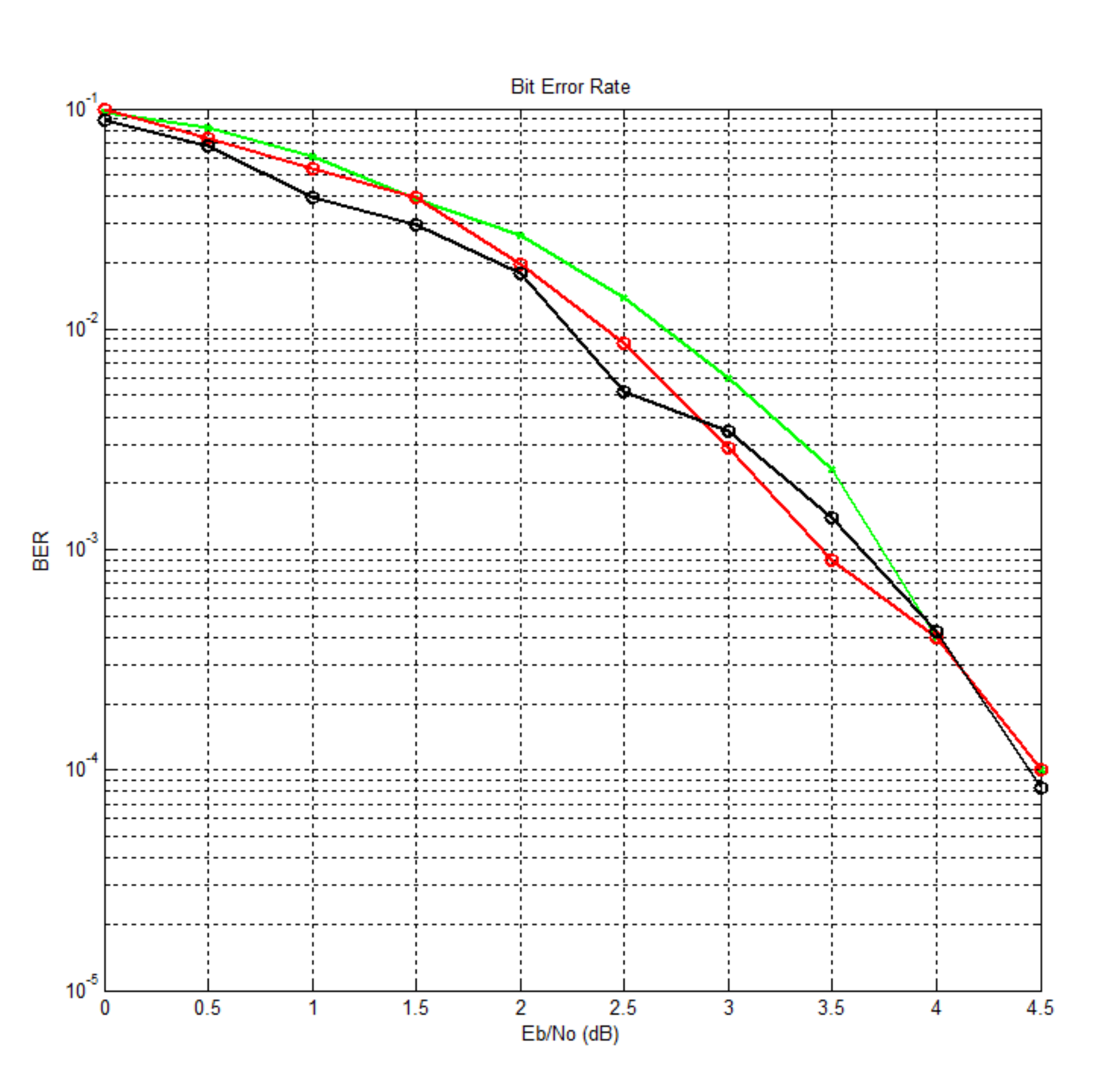}
    \caption{Bit error rate for $[2401,1715]$ code (green) corresponding to graph $D(6, 7)$,
$[625, 375]$ code (red) corresponding to graph $D(8, 5)$ and $[243, 81]$ code (black) corresponding to graph $D(10, 3)$}
        \label{fig:3}
    \end{center}
\end{figure}

Fig. \ref{fig:3} presents bit error rate for example LDPC codes corresponding to graphs $D(n,q)$. Green and red codes have similar code rate like codes based on graphs $F$. Comparing figures Fig. \ref{fig:1} and Fig. \ref{fig:2} with Fig. \ref{fig:3} we can see that our codes perform very well and have better error correcting properties.

\section{Conclusion}

The present work shows a modification of the family of graphs introduced by Ustimenko and Woldar by changing the 
finite field to a finite ring. LDPC codes are constructed based on the original family and on the modified one. 
An analysis is performed on such codes as well as a BER simulation. A comparison of the results with former graph 
based codes shows an improvement of the error correcting properties with the graphs introduced here.
It's also important to mention that they can easily be combined and work with existing decoding techniques.


\nocite{*}
\bibliographystyle{amsplain}

\bibliography{2016-4}{}




\end{document}